\documentclass[aps,pra,twocolumn,amsfonts,amssymb,amsmath,showpacs,
floatfix,nofootinbib,citesort]{revtex4-2}
\usepackage{mathrsfs}
\usepackage{amsfonts}
\usepackage{amstext}
\usepackage{amsmath}
\usepackage{amssymb}
\usepackage{bm}
\usepackage{CJK}
\usepackage{bbm}
\usepackage[dvips]{graphicx}
\def\qed{\leavevmode\unskip\penalty9999 \hbox{}\nobreak\hfill
	\quad\hbox{\leavevmode  \hbox to.77778em{%
			\hfil\vrule   \vbox to.675em%
			{\hrule width.6em\vfil\hrule}\vrule\hfil}}
	\par\vskip3pt}

\usepackage{amssymb}
\usepackage{graphicx}
\usepackage{graphics}
\usepackage{amsmath}
\usepackage{amsthm}
\usepackage{color}
\usepackage{dsfont}
\usepackage{textcomp}
\definecolor{darkred}  {rgb}{0.5,0,0}
\definecolor{darkblue} {rgb}{0,0,0.5}
\definecolor{darkgreen}{rgb}{0,0.5,0}
\usepackage{hyperref}
\hypersetup{
	pdftitle = {QRT Proposal},
	pdfauthor = {},
	colorlinks = true,
	urlcolor  = blue,         
	linkcolor = red,     
	citecolor = blue,    
	filecolor = darkred       
}
\usepackage{mathtools}
\def\ra{\rangle}
\def\la{\langle}

\def\ot{\otimes}
\newtheorem{theorem}{Theorem}

\newtheorem{cor}[theorem]{Corollary}

\newcommand{\bea}{\begin{eqnarray}}
	\newcommand{\eea}{\end{eqnarray}}
\newcommand{\be}{\begin{equation}}
	\newcommand{\ee}{\end{equation}}
\newcommand{\ba}{\begin{equation}\begin{aligned}}
		\newcommand{\ea}{\end{aligned}\end{equation}}

\newcommand{\beax}{\begin{eqnarray*}}
	\newcommand{\eeax}{\end{eqnarray*}}
\newcommand{\bex}{\begin{equation*}}
	\newcommand{\eex}{\end{equation*}}

\theoremstyle{remark}

\newtheorem{example}{Example}

\def\be{\begin{equation}}
	\def\ee{\end{equation}}

\newcommand{\mC}{\mathcal{C}}

\newcommand{\mH}{\mathcal{H}}

\newcommand{\mB}{\mathcal{B}}
\newcommand{\mP}{\mathcal{P}}

\newcommand{\mS}{\mathcal{S}}

\newcommand{\lr}{\rangle\langle}

\newcommand{\tr}{{\rm Tr}}

\newcommand{\rmi}{{\rm i}}

\begin{document}
	

\preprint{APS/123-QED}
\begin{CJK*}{GB}{gbsn}
	\title{Computable lower bound of the parameterized entanglement monotone\\}
		

		\author{Ning Yang}
		\author{Yu Guo}
		\email{guoyu3@aliyun.com}
		
		\affiliation{School of Mathematical Sciences, Inner Mongolia University, Hohhot, Inner Mongolia 010021, People's Republic of China}
		\affiliation{Inner Mongolia Key Laboratory of Mathematical Modeling and Scientific Computing, Inner Mongolia University, Hohhot, Inner Mongolia 010021, People's Republic of China}

		\author{Shuanping Du}
		\affiliation{School of Mathematical Sciences, Xiamen University, Xiamen, Fujian, 361000, People's Republic of China}


		
\begin{abstract}
			
Although numerous measures of entanglement have been proposed so far, the calculation of a given faithful entanglement measure is a hard work since it is always involved in some optimization process. It is, therefore, important to estimate the lower bound of a given entanglement measure for an arbitrary quantum state. This results in a subject of intensive mathematical research. In particular, along this line, the lower bounds of concurrence or other measures that are induced from concurrence have been explored a lot. Here, we investigate the lower bounds of two kinds of entanglement monotones, i.e., $q$-concurrence ($q>1$) and $\alpha$-concurrence ($0<\alpha<1$), or termed the parameterized entanglement monotone together. We obtain, in the light of the informationally complete ($N$, $M$)-positive operator-valued measure [($N$, $M$)-POVM], the lower bounds for the case of $\frac12<\alpha<1$, $1<q<2$ for two-qudit states, and the case of $2\leqslant q<3$ for two-qubit states. We list several examples which show that the lower bounds based on ($N$, $M$)-POVM outperform that of GSIC-POVM and SIC-POVM, and all these measurement based bounds are better then the ones induced by positive partial transpose (PPT) and realignment criteria in literature. In addition, we obtain an analytical formula of the parameterized entanglement monotone with $\frac12<\alpha<1$ and $1<q<2$ for the isotropic state.

\end{abstract}
		
		
		\maketitle
	\end{CJK*}


\section{Introduction}

Up until now, more than 40 bipartite measures of entanglement have been proposed in the past three decades~\cite{Guo2025arxiv}. Among these measures, one of the most popular measures is the concurrence due to the fact that it is easy to calculate for pure states and it has the analytic formula for two-qubit quantum states~\cite{Wootters1998prl}, and it is closely related to the entanglement of formation~\cite{Hill1997prl}. Thereby, a list of measures induced from concurrence was explored, such as the tangle~\cite{Rungta2003pra}, $q$-concurrence~\cite{Yang2021pra,Bao2025qip}, $\alpha$-concurrence~\cite{Wei2022jpamt,Bao2025qip}, total concurrence~\cite{Xuan2025aqt}, etc. All these measures can be evaluated in terms of the PPT and realignment criteria~\cite{Yang2021pra,Wei2022qip,Bao2025qip,Chen2005prl,Wei2022jpamt,Xuan2025aqt}.

Recently, further lower bounds of concurrence in terms of the informationally complete ($N$, $M$)-POVM~\cite{Siudzinska2022} was proposed in~\cite{Wang2025pra}. It was shown that there exist states such that these bounds are tighter than some existing lower bounds in Ref.~\cite{Shi2023qip,Li2024pra}. Recalling that, the symmetric informationally complete positive operator-valued measure (SIC-POVM)~\cite{Renes2004} may not exist, but general symmetric informationally complete positive operator-valued measure (GSIC-POVM)~\cite{Appleby2007os,Gour2014} does exist for any quantum system~\cite{Gour2014}. The informationally complete ($N$, $M$)-POVM is a generalization of GSIC-POVM, recently proposed by Siudzi\'nska~\cite{Siudzinska2022}. These POVMs can be used to detect entanglement~\cite{Chenbin2015qip,Lu2025pra,Shang2018pra,Qixiaofei2025jpa} and to investigate relationships between entropy and entanglement~\cite{Huang2023ps,Rastegin2013epjd}. Notably, these separability criteria based on quantum measurements are of particular significance since they are experimentally implementable.

The lower bounds of the parameterized entanglement monotone were explored recently by virtue of the PPT and realignment criteria. A lower bound of $q$-concurrence was presented for the case of $q \geqslant2$ in~\cite{Yang2021pra}, and soon after, an improved lower bound for $q$-concurrence was obtained~\cite{Wei2022qip} for $q\geqslant2$. The case of $0 \leqslant \alpha \leqslant 1/2$ was discussed in~\cite{Wei2022jpamt}. Very recently, the cases of $1/2< \alpha < 1$ and $1 < q < 2$ were explored~\cite{Bao2025qip}. Inspired by the method in~\cite{Wang2025pra}, in this paper, we discuss the lower bounds of the parameterized entanglement monotone based on the informationally complete ($N$, $M$)-POVM. Consequently, we give several examples which show that our bounds have advantage over those of the literature above.

The paper is organized as follows. We start by reviewing the parameterized entanglement monotone, the informationally complete ($N$, $M$)-POVM and listing the earlier results for the lower bounds of the parameterized entanglement monotone. In Sec.~\ref{sec3}, we propose the main results: By means of the informationally complete ($N$, $M$)-POVM, the lower bounds of the parameterized entanglement monotone of the states acting on $\mH^{AB}$ with $\dim\mH^A=\dim\mH^B=d\geqslant 2$ for the case of $1<q<2$, $\frac12<\alpha<1$, and the lower bounds of the two-qubit state for the case of $2\leqslant q<3$ are derived, respectively. The advantages of our bounds are illustrated with examples in Sec.~\ref{sec4}. We conclude with some discussions in
Sec.~\ref{sec5}. In Appendix~\ref{sec7} we give, for the cases of $1/2<\alpha<1$ and $1<q<2$, the proof of analytic formula of the parameterized entanglement monotone for the isotropic state (i.e., Lemma in Sec.~\ref{sec4}) and in Appendix~\ref{sec8} we list the ($N$, $M$)-POVMs we used in the examples in Sec.~\ref{sec4}.

\section{Preliminaries}

\subsection{The existing lower bounds of the parameterized entanglement monotone}

We review in this subsection the definition of the parameterized entanglement monotone and the existing lower bounds. 
For an arbitrary bipartite pure state $|\psi\ra$ in the state space described by the Hilbert space $\mH^{AB}:=\mH^A\ot \mH^B$, the parameterized entanglement monotone is defined by~\cite{Yang2021pra,Bao2025qip}
\bea\label{eq1_C_q}
C_q(|\psi\ra)=1-\tr\rho_A^q, \quad q>1,
\eea
where $\rho_A$ is the reduced state of $|\psi\ra$.
It is extended to mixed states $\rho$ via the convex-roof extension
\bea
C_q(\rho)=\min_{\{p_i,|\psi_i\ra\}}\sum_ip_iC_q(|\psi\ra),
\eea
where the minimum is taken over all pure state decompositions of $\rho=\sum\limits_{j}p_j|\psi_j\lr\psi_j|$, 
$C_q$ is also called $q$-concurrence~\cite{Yang2021pra,Bao2025qip}.
Alternatively, for the case of $0<\alpha<1$, the parameterized entanglement monotone is defined by~\cite{Wei2022jpamt,Bao2025qip}
\bea\label{eq1_C_alpha}
C_{\alpha}(|\psi\ra)=\tr\rho_A^{\alpha}-1, \quad 0<\alpha<1,
\eea
and for mixed state, it is defined by the convex-roof extension as well.
In this case, $C_{\alpha}$ is also called $\alpha$-concurrence~\cite{Wei2022jpamt}.

For any entanglement state $\rho$ acting on $\mH^{AB}$ with $\dim\mH^A=d_A$ and $\dim\mH^B=d_B$ ($d_A \leq d_B$), the $q$-concurrence $C_q(\rho)$ with $q\geqslant 2$ has the following lower bound~\cite{Yang2021pra}
\bea\label{inq4}
C_q(\rho)\geq\frac{\left( \max\left\lbrace \left\| \rho^{T_a}\right\|_{\tr}^{q-1},\left\|\rho^R\right\|_{\tr}^{q-1}\right\rbrace -1\right) ^2}{d^{2q-2}-d^{q-1}},
\eea
where $\rho^{T_a}$ is the partial transposition of $\rho$ up to part $A$~\cite{Horodecki1996pla,Peres1996prl}, $\rho^{R}$ is the realignment of $\rho$~\cite{Chen2003qic,Rudolph2004lmp}, that is, if $\rho=\sum_{i,j,k,l}a_{i,j,k,l}|i\ra\la j|^A\ot|k\ra\la l|^B$ up to some basis $\{|i\ra^A|k\ra^B\}$ of $\mH^{AB}$, then 
\beax
\rho^{T_a}&=&\sum_{i,j,k,l}a_{i,j,k,l}\left( |i\ra\la j|^A\right)^T\ot|k\ra\la l|^B,\label{partial-transpose}\\
\rho^{R}&=&\sum_{i,j,k,l}a_{i,j,k,l}|j\ra^A|i\ra^A\la l|^B\la k|^B,\label{realignment}
\eeax
$X^T$ denotes the transpose of $X$ up to some given orthonormal basis,
$\|A\|_{\tr}$ denotes the trace norm of $A$, i.e.,
$\|A\|_{\tr}=\tr|A|$, $|A|=(A^\dag A)^{1/2}$, $d=d_A$.
Subsequently, for the case of $d = 2$ and $s\leqslant q < 3$ with $s=2.4721$, Wei et al.~\cite{Wei2022qip} provided the bound
\be\label{ins_q_3}
C_q(\rho) \geqslant \frac{1 - 2^{1-q}}{2 - 2^{2-s}} \left( \max\left\{ \|\rho^{T_a}\|_{\tr}, \|\rho^{R}\|_{\tr} \right\} - 1 \right)^2.
\ee
Meanwhile, they derived another improved lower bound
\be\label{inq_2}
C_q(\rho) \geqslant \frac{1 - d^{1-q}}{(d - 1)^2} \left( \max\left\{ \|\rho^{T_a}\|_{\tr}, \|\rho^{R}\|_{\tr} \right\} - 1 \right)^2,
\ee
which holds for either $q \geqslant 2$ and $d \geqslant 3$ or $q \geqslant 3$ and $d = 2$. 
A lower bound for $q$-concurrence with $1 < q < 2$ is~\cite{Bao2025qip}
\begin{widetext}
\bea\label{eq13}
C_q(\rho)
 \geqslant 1 - \left[ 1 - \frac{1 - d^{\frac{(1-q)p}{p-1}}}{(d - 1)^2} \left( \max\left\{ \|\rho^{T_a}\|_{\tr}, \|\rho^{R}\|_{\tr} \right\} - 1 \right)^2 \right]^{\frac{p-1}{p}},
\eea
where $1 < p \leq \frac{3}{4 - q}$.

For $0<\alpha\leqslant 1/2$, Wei et al.~\cite{Wei2022jpamt} proposed the following lower bound 
\be\label{eq12}
\begin{aligned}
C_\alpha(\rho) &\geqslant \frac{d^{1-\alpha} - 1}{d - 1} \left( \max\left\{ \|\rho^{T_a}\|_{\tr}, \|\rho^{R}\|_{\tr} \right\} - 1 \right).
\end{aligned}
\ee
For the case $1/2 < \alpha < 1$, a lower bound for $\alpha$-concurrence is~\cite{Bao2025qip}
\be\label{eq14}
C_\alpha(\rho)\geqslant d^{\frac{2s(1-\alpha)-\alpha}{2s}} \left( \frac{\sqrt{d} + \max\left\{ \|\rho^{T_a}\|_{\tr}, \|\rho^{R}\|_{\tr} \right\}}{\sqrt{d} + 1} \right)^{\frac{\alpha}{s}} - 1,
\ee
\end{widetext}
where $s \in (0, 1/2)$.

\subsection{($N$, $M$)-PVOM}

Quantum measurement is realized by positive operator-valued
measure (POVM), which is a basic tool in quantum information processing since it is indispensable when we extract information from the given quantum system. A set of operators $\{E_i\}_{i=1}^n$ acting on the state space $\mH$ is called a POVM if $E_i\geqslant 0$ for each $i$ and $\sum_{i=1}^{n} E_i = I$, where $I$ denotes the identity operator on $\mH$. Hereafter in this subsection we always assume that the dimension of the state space is $d$. If $\{E_r = |\phi_r\ra\la \phi_r|/d|r=1, 2, \cdots, d^2\}$ satisfies 
\beax
|\la\phi_r| \phi_s\ra|^2 = \frac{1}{d + 1}
\eeax
for all $r\neq s$, it is called a SIC POVM on $\mH$ with $\dim\mH=d$. However, these two classes of POVMs may not exist for the Hilbert space with some fixed dimension(s). Consequently, the general symmetric
informationally complete positive operator-valued measure
(GSIC POVM)~\cite{Appleby2007os,Gour2014} is proposed: A POVM that consists of $d^2$ positive-semidefinite operators $\{P_\alpha\}^{d^2}_{\alpha=1}$ acting on $\mH$ is called a general SIC POVM if $\tr\left(P_\alpha^2\right) = \tr\left(P_\beta^2\right) \neq \frac{1}{d^3}$ for all $\alpha, \beta \in \{1, 2, \ldots, d^2\}$, and $\tr\left(P_\alpha P_\beta\right) = \tr\left(P_{\alpha'} P_{\beta'}\right)$ for all $\alpha \neq \beta$ and $\alpha' \neq \beta'$. GSIC POVM exists for state space with any dimension and it recovers to SIC POVM whenever $\tr\left(P_\alpha^2\right)=\frac{1}{d^2}$.

Going further, a wider range of classes of POVM, termed ($N$, $M$)-POVM~\cite{Siudzinska2022}, was introduced. An ($N$, $M$)-POVM is defined as a set of $N$ POVMs $\{E_{\alpha, k} \mid k=1 ,2, \ldots, M\}$ $(\alpha=1, 2, \ldots, N)$ that satisfies the symmetry condition, i.e.,
\bea
\tr(E_{\alpha, k})&=&w,\\\notag
\tr(E_{\alpha, k}^2)&=&x,\\\notag
\tr(E_{\alpha,k} E_{\alpha,l})&=&y, \quad l \neq k,\\\notag
\tr(E_{\alpha,k}E_{\beta,l})&=&z, \quad \beta \neq \alpha,
\eea
where $$w=\frac{d}{M}, \ y=\frac{d-M x}{M(M-1)}, \ z=\frac{d}{M^2},$$ and the parameter $x$ satisfies $\frac{d}{M^2}<x\leqslant \min\left\{\frac{d^2}{M^2},\frac{d}{M}\right\}$.
The ($N$, $M$)-POVM is called informationally complete if and only if $N(M - 1) = d^2 - 1$. For the case of $N=1$ and $M=d^2$, it reduces to GSIC-POVM while the case of $N=d+1$ and $M=d$ is reduced to mutually unbiased measurements~(MUMs)  .

There is a general method of constructing an informationally complete ($N$,$M$)-POVM. 
Let $\{G_0=I_d/\sqrt{d}, G_{\alpha,k}\mid\alpha=1,\ldots,N;k=1,\ldots,M-1\}$ with $\tr G_{\alpha,k}=0$ be an orthonormal Hermitian operator basis of the Hilbert space $\mB(\mH)$, i.e., the space of all linear operators on $\mH$ with the inner product induced by $\la A, B\ra=\tr(A^\dag B)$ for any $A$, $B\in\mB(\mH)$. Then the set of operators defined by
\bea\label{NM-POVM}
E_{\alpha,k}= \frac{1}{M}I_d+tH_{\alpha,k}
\eea
forms an informationally complete ($N$, $M$)-POVM, where
\beax
H_{\alpha,k}=
\begin{cases}
G_{\alpha}-\sqrt{M}(\sqrt{M}+1)G_{\alpha,k}, & k=1,\ldots,M-1, \\
(\sqrt{M}+1)G_{\alpha}, & k=M
\end{cases}
\eeax
with $G_{\alpha}=\sum_{k=1}^{M-1}G_{\alpha,k}$, the parameter $t$ should be chosen such that $E_{\alpha,k}\geqslant0$ and
\beax
-\frac{1}{M}\frac{1}{\lambda_{\max}} \leqslant t \leqslant\frac{1}{M} \frac{1}{|\lambda_{\min}|},
\eeax
where $\lambda_{\max}$ and $\lambda_{\min}$ are the maximal and minimal eigenvalues from all eigenvalues of $H_{\alpha, k}$s, respectively, and the parameters also $t$ and $x$ satisfy
\beax
x = \frac{d}{M^2} + t^2(M - 1)(\sqrt{M} + 1)^2.
\eeax

\section{Lower bounds of the parameterized entanglement monotone}\label{sec3}

We denote the set of all states of the quantum system described by the Hilbert space $\mH$ by $\mS(\mH)$ and by $\mS^{AB}$ the set $\mS(\mH^{AB})$.

In this section, we discuss the lower bound of the parameterized entanglement monotone for the states in $\mS^{AB}$ with $\dim\mH^A=\dim\mH^B=d$ by means of the informationally complete ($N$, $M$)-POVM. Let $\{E_{\alpha, k}|\alpha=1,\cdots,N; k=1,\cdots, M\}$ be an informationally complete ($N$, $M$)-POVM defined as in Eq.~\eqref{NM-POVM} with $\dim\mH=d$, and let $\{|e_{\alpha, k}\ra\}\ (\alpha=1, \cdots,N; k=1, \cdots, M)$ be an orthonormal basis of $\mathbb{C}^{NM}$. In what follows, we write
\be\label{eqsz}
\mathcal{P}(\rho):=\sum_{\alpha, \beta=1}^{N}\sum_{k,l=1}^{M}\tr\left[\rho(E_{\alpha,k} \ot E_{\beta,l})\right]|e_{\alpha,k}\ra\la e_{\beta,l}|.
\ee

\subsection{Lower bounds of the $q$-concurrence}

We now consider the lower bounds of the $q$-concurrence in terms of Eq.~\eqref{eqsz} based on the informationally complete ($N$, $M$)-POVM.

\begin{theorem}\label{th1}
For any bipartite state $\rho\in\mS^{AB}$ with $\dim\mH^A=\dim\mH^B=d$, and using the notations established above, for $1<q<2$, we have
\be\label{in30_5}
C_q(\rho)\geqslant\frac{1-d^{1-q}}{a^2(d-1)^2}\left(\|\mathcal{P}(\rho)\|_{\tr}-b-a\right)^2,
\ee
where $a=\frac{xM^2-d}{M(M-1)}$, $b=\frac{d^3-xM^2}{dM(M-1)}$.
\end{theorem}

\begin{proof}
Let 
\beax
f(\vec{x})=\frac{1-\sum_{i=1}^{d}x_{i}^{q}}{\left[\left(\sum_{i=1}^{d}\sqrt{x_{i}}\right )^{2}-1\right]^2},
\eeax 
where $\vec{x}=(x_1,x_2,\cdots,x_d)$ with $x_1\geqslant x_2\geqslant \cdots\geqslant x_d\geqslant 0$, $ \sum_{i} x_i = 1 $, and $x_1 \neq1$. To find the minimum value of $f(\vec{x})$ with $1<q<2$ and $d\geqslant2$, we employ the method of Lagrange multipliers under the constraints $\sum_{i} x_i = 1$ and  $x_i \geqslant 0$. Under these constraints, there is only one stable point which occurs when each $ x_i =\frac{1}{d}$ ($i = 1, \cdots, d$). Notably, at the stable point, the value of $f(\vec{x})$ is $\frac{1 - d^{1-q}}{(d - 1)^2}$. Let $p_0$ be the stable point. Then the second partial derivative of $f(\vec{x})$ at $p_0$ with respect to $x_i$ and $x_j$ ($i \neq j$) are
\beax
c=\frac{\partial^2f}{\partial x_i^2}\bigg|_{p_0}=&&\frac{q(1-q)d^{2-q}+d(1-d^{1-q})}{(d-1)^2}\\
&&+\frac{4qd^{2-q}(d-1)+6d^2(1-d^{1-q})}{(d-1)^4},
\eeax
\beax
c'=\frac{\partial^2 f}{\partial x_i \partial x_j}\bigg|_{p_0}
=&&\frac{(1-d^{1-q})[6d^2-d(d-1)]}{(d-1)^4}\\
&&+\frac{4q(d-1)d^{2-q}}{(d-1)^4}.
\eeax
Therefore the Hessian matrix of $f(\vec{x})$ at point $p_0$ is
\beax
H = \begin{bmatrix}
	c & c' & c' & \dots & c' \\
	c' & c & c' & \dots & c' \\
	c' & c' & c & \dots & c' \\
	\vdots & \vdots & \vdots & \ddots & \vdots \\
	c' & c' & c' & \dots & c
	\end{bmatrix}_{d \times d}
\eeax
and its eigenvalues are respectively $c-c'$ (with multiplicity $d-1$) and $c+(d-1)c'$ (with multiplicity $1$). A simple verification shows that
\beax
c-c'>0,\ \ c+(d-1)c'>0,
\eeax
when $1 < q < 2$ and $d \geqslant 2$. This implies that $p_0$ is a minimum extreme value point, which is just the minimum value point of $f(\vec{x})$. It follows that for $1 < q < 2$ and $d \geqslant 2$, 
\beax
1-\sum_{i=1}^{d}x_{i}^{q}\geqslant\frac{1 - d^{1-q}}{(d - 1)^2}\left[\left(\sum_{i=1}^{d}\sqrt{x_{i}}\right)^{2}-1\right]^2.
\eeax
Let $|\psi\ra $ be a pure state in $\mH^{AB}$ with Schmidt coefficients $\{\lambda_i\}$, and we set $\lambda_i=\sqrt{x_i}$, $i=1,2,\cdots d$. On the other hand, let $a=\frac{x M^{2}-d}{M(M-1)}$ and $b=\frac{d^{3}-x M^{2}}{d M(M-1)}$ with $a\neq0$. According to Theorem 1 in~\cite{Wang2025pra}, 
\beax\label{eq2}
&&\|\mathcal{P}(|\psi\ra\la\psi|)\|_{\tr}\\
&=&\sum_{i,j} \lambda_{i} \lambda_{j} \left[\frac{x M^{2}-d}{M(M-1)}+\frac{d^{3}-x M^{2}}{dM(M-1)} \delta_{i j}\right] \\
&=&a\sum_{i,j} \lambda_{i} \lambda_{j}+b\sum_{i} \lambda_{i}^2=a \left(\sum_{i} \lambda_{i}\right)^2+b.
\eeax
That is
\bea \label{purestate-eq} \frac{1}{a}\left(\|\mathcal{P}(\rho)\|_{\tr}-b\right)=\|\rho^{T_a}\|_{\tr}=\|\rho^{R}\|_{\tr}=\left(\sum_{i}\lambda_{i}\right)^2~~~~
\eea  
holds for any bipartite pure state $\rho=|\psi\ra\la\psi|\in\mS^{AB}$. Thus, we have
\be\label{eq16}
C_q(|\psi\ra)\geqslant\frac{1-d^{1-q}}{a^2(d-1)^2}\left(\|\mathcal{P}(|\psi\ra\la\psi|)\|_{\tr}-b-a\right)^2.
\ee	
Next, we demonstrate that the above result is also valid for mixed states. For an arbitrary mixed state $\rho$, consider its optimal decomposition $\rho=\sum_{i}p_i|\psi_i\ra\la\psi_i|$, which minimizes the average $q$-concurrence. Since $f(x)=x^2$ is convex, from Eq.~\eqref{eq16}, it follows that
\beax
C_q(\rho)&=&\sum_{i}p_iC_q(|\psi_i\ra\la\psi_i|)\\
&\geqslant& \sum_{i}p_i\left[\frac{1 - d^{1-q}}{a^2(d - 1)^2} \left(\|\mathcal{P}(|\psi_i\ra\la\psi_i|)\|_{\tr}-b- a \right)^2\right]\\
&\geqslant&\frac{1 - d^{1-q}}{a^2(d - 1)^2} \left( \sum_{i}p_i\|\mathcal{P}(|\psi_i\ra\la\psi_i|)\|_{\tr}-b- a \right)^2.
\eeax
By the convexity of the trace norm and the linearity of $\mathcal{P}$, we have $\|\mathcal{P}(\rho)\|_{\tr}=\|\sum_{i}p_i\mathcal{P}(|\psi_i\ra\la\psi_i|)\|_{\tr}\leqslant \sum_{i}p_i\|\mathcal{P}(|\psi_i\ra\la\psi_i|)\|_{\tr}$, which completes the proof for mixed states.  
\end{proof}

For the two-qubit case, we can obtain the following lower bound.

\begin{theorem}\label{th2}
For any two-qubit state $\rho$, with the notations of Theorem~\ref{th1} and $2\leqslant q<3$, we have 
\be\label{in25_7}
C_q(\rho)\geqslant1-\left[1- \frac{1-d^{-q}}{a^2(d-1)^2} (\|\mathcal{P}(\rho)\|_{\tr}-b-a)^2\right]^{\frac{q-1}{q}}.
\ee
\end{theorem}

\begin{proof}
Let $|\psi\ra$ be a two-qubit pure state with Schmidt coefficients $\{\lambda_i\}$. By the H\"{o}lder inequality, we have
\beax
\sum_{i}\lambda_{i}^{2q}&=&\sum_{i}\lambda_{i}^{\frac{2}{q}}\lambda_{i}^{2(q-\frac{1}{q})}\\
&\leqslant&\left[\sum_{i}\left(\lambda_{i}^{\frac{2}{q}}\right)^{q}\right]^{\frac{1}{q}}\left\{\sum_{i}\left[\lambda_{i}^{2\left(q-\frac{1}{q}\right)}\right]^{\frac{q}{q-1}}\right\}^{\frac{q-1}{q}}\\
&=&\left[\sum_{i}\lambda_{i}^{2(q+1)}\right]^{\frac{q-1}{q}}.
\eeax
Since $q+1\geqslant3$, from Eqs.~\eqref{inq_2} and~\eqref{purestate-eq} we obtain
\beax\label{in27}
&&C_q(|\psi\ra)\\
&=&1-\sum_{i}\lambda_{i}^{2q}\geqslant1-\left[\sum_{i}\lambda_{i}^{2(q+1)}\right]^{\frac{q-1}{q}}\\
&\geqslant&1-\left[1-\frac{1-d^{-q}}{a^2(d-1)^2}(\|\mathcal{P}(|\psi\ra\la\psi|)\|_{\tr}-b-a)^2\right]^{\frac{q-1}{q}}.
\eeax
For any given mixed state $\rho$, let $\rho=\sum_{i}p_i|\psi_i\ra\la\psi_i|$ be the optimal decomposition satisfying $C_q(\rho)=\sum_{i}p_iC_q(|\psi_i\ra\la\psi_i|)$. Note that the function $f(x)=x^{\frac{q-1}{q}}$ is concave and $g(x)=x^2$ is convex. Thus
\begin{widetext}
\beax
C_q(\rho)&=&\sum_{i}p_iC_q(|\psi_i\ra\la\psi_i|)\geqslant 1-\sum_{i}p_i\left[1- \frac{1-d^{-q}}{a^2(d-1)^2} (\|\mathcal{P}(|\psi_i\ra\la\psi_i|)\|_{\tr}-b-a)^2\right]^{\frac{q-1}{q}}\\
&\geqslant& 1-\left[1- \frac{1-d^{-q}}{a^2(d-1)^2} \sum_{i}p_i\left(\|\mathcal{P}(|\psi_i\ra\la\psi_i|)\|_{\tr}-b-a\right)^2\right]^{\frac{q-1}{q}}\\
&\geqslant& 1-\left[1- \frac{1-d^{-q}}{a^2(d-1)^2} (\sum_{i}p_i\|\mathcal{P}(|\psi_i\ra\la\psi_i|)\|_{\tr}-b-a)^2\right]^{\frac{q-1}{q}}\ \geqslant 1-\left[1- \frac{1-d^{-q}}{a^2(d-1)^2} (\|\mathcal{P}(\rho)\|_{\tr}-b-a)^2\right]^{\frac{q-1}{q}}
\eeax
\end{widetext}
as desired. This completes the proof.
\end{proof}

By Eqs.~\eqref{ins_q_3}-\eqref{eq13}, together with Eq.~\eqref{purestate-eq}, we can get the following corollary.

\begin{cor}\label{cor1}
With the above notations, for either $q\geqslant3$ with $d=2$ or $q\geqslant2$ with $d\geqslant3$,
\be\label{in25_q_3}
C_q(\rho)\geqslant\frac{1-d^{1-q}}{a^2(d-1)^2}\left(\|\mathcal{P}(\rho)\|_{\tr}-b-a\right)^2,
\ee
and	for $d=2$ and $s\leqslant q<3$ with $s=2.4721$,
\be\label{in24_6}
C_q(\rho)\geqslant \frac{1-2^{1-q}}{a^2(2-2^{2-s})}\left(\|\mathcal{P}(\rho)\|_{\tr}-b-a\right)^2.
\ee
For $1<q<2$ with $d\geqslant2$ and $1<p\leqslant\frac{2}{3-q}$, 
\be\label{in29_1_q_2}
C_q(\rho)\geqslant1-\left[1-\frac{1-d^{\frac{(1-q)p}{p-1}}}{a^2(d-1)^2}\left(\|\mathcal{P}(\rho)\|_{\tr}-b-a \right)^2\right]^{\frac{p-1}{p}}.
\ee
\end{cor}

\subsection{Lower bounds of the $\alpha$-concurrence}

For the $\alpha$-concurrence, we have the following lower bound.

\begin{theorem}\label{th4}
For any bipartite state $\rho\in\mS^{AB}$ with $\dim\mH^A=\dim\mH^B=d$, with the notations of Theorem~\ref{th1}, the $\alpha$-concurrence satisfies
\bea\label{in2}
C_\alpha(\rho)\geqslant \frac{\|\mathcal{P}(\rho)\|_{\tr}-b}{ad^\alpha}-1,\quad 0<\alpha<1.
\eea
\end{theorem}

\begin{proof}
For any pure state $|\psi\ra$ with Schmidt coefficients $\{\lambda_i\}$, 
by Cauchy-Schwartz inequality, 
\beax
\left(\sum_{i}\lambda_{i}\right)^2=\left(\sum_{i}\lambda_{i}^\alpha\lambda_{i}^{1-\alpha}\right)^2\leqslant\left(\sum_{i}\lambda_{i}^{2\alpha}\right)\left[\sum_{i}\lambda_{i}^{2(1-\alpha)}\right],
\eeax
which implies
\beax
\sum_{i}\lambda_{i}^{2\alpha}\geqslant\frac{\left(\sum_{i}\lambda_{i}\right)^2}{\sum_{i}\lambda_{i}^{2(1-\alpha)}}.
\eeax
Since $f(x)=x^{1-\alpha}$ is concave, it follows that $\sum_{i}\lambda_{i}^{2(1-\alpha)}\leqslant d^\alpha$.
Hence, by Eq.~\eqref{purestate-eq}, for $0<\alpha<1$, $\alpha$-concurrence satisfies
\bea\label{in1}
C_\alpha(|\psi\ra)&=&\sum_{i}\lambda_{i}^{2\alpha}-1\geqslant\frac{\left(\sum_{i}\lambda_{i}\right)^2}{d^\alpha}-1\nonumber\\
&=& \frac{\|\mathcal{P}(|\psi\ra\la\psi|)\|_{\tr}-b}{ad^\alpha}-1.
\eea
Thus, we can obtain Eq.~\eqref{in2} for pure state, and it is straightforward to verify that the same conclusions hold for mixed states as well as that of Theorem~\ref{th1}.
\end{proof}

By Eq.~\eqref{purestate-eq}, together with Eqs.~\eqref{eq12} and~\eqref{eq14}, the lower bound of the $\alpha$-concurrence can be stated as the following corollary.

\begin{cor}\label{cor2}
For any bipartite state $\rho\in\mS^{AB}$ with $\dim\mH^A=\dim\mH^B=d$, with the notations of Theorem~\ref{th1}, we have
\bea\label{in21_1}
C_\alpha(\rho) \geqslant \frac{d^{1-\alpha}-1}{a(d-1)}\left(\|\mathcal{P}(\rho)\|_{\tr}-b-a\right),
\eea
whenever $0<\alpha\leqslant1/2$, and 
\be\label{in22_2}
C_\alpha(\rho) \geqslant  d^{\frac{2s(1-\alpha)-\alpha}{2s}}\left[\frac{a\sqrt{d} + \left(\|\mathcal{P}(\rho)\|_{\tr}-b\right)}{a(\sqrt{d}+1)}\right]^{\frac{\alpha}{s}}-1, 	
\ee
whenever $1/2<\alpha<1$, where $0<s<1/2$.
\end{cor}

We remark here that the lower bound in Eq.~\eqref{in21_1} is larger than the bound in Eq.~\eqref{in2} for any state whenever $0<\alpha<\frac12$. In fact, 
for any pure state $|\psi\ra$, from~\eqref{in1} and Theorem 2 of Ref.~\cite{Wei2022jpamt}, we have
\beax
C_\alpha(|\psi\ra)&\geqslant&\frac{\left(\sum_{i}\lambda_{i}\right)^2}{d^\alpha}-1,\label{22}\\
C_\alpha(|\psi\ra)&\geqslant&\frac{d^{1-\alpha}-1}{d-1}\left[\left(\sum_{i}\lambda_{i}\right)^2-1\right].\label{23}
\eeax
One can easily check  \beax \frac{d^{1-\alpha}-1}{d-1}\left[\left(\sum_{i}\lambda_{i}\right)^2-1\right]\geqslant\frac{\left(\sum_{i}\lambda_{i}\right)^2}{d^\alpha}-1
\eeax
whenever $0<\alpha<\frac12$, which implies the lower bound in Eq.~\eqref{in21_1} is tighter than that of Eq.~\eqref{in2}.

\begin{table*}[ht]
\caption{\label{tab:table1} The lower bounds of the parameterized entanglement monotone induced from $\mP^\sharp$. For convenience, with some abuse of notations, we denote them by $L^\sharp_{\alpha, d}$, $L_{\alpha,s,d}^\sharp$,  $L_{q,p,d}^\sharp$, $L_{q,s,2}^\sharp$, $\bar{L}_{\alpha,s,d}^\sharp$,  $\bar{L}_{q,p,d}^\sharp$, $\bar{L}_{q,s,2}^\sharp$, and $L_{q,d}^\sharp$, respectively.}	
\begin{ruledtabular}
\renewcommand{\arraystretch}{1.3}
\begin{tabular}{ll}			
OLB     &	 ILB     \\ 
\hline
    $L^{\sharp}_{\alpha, d}=\frac{d^{1-\alpha}-1}{a(d-1)}\left(\|\mathcal{P}^{\sharp}(\rho)\|_{\tr}-b^{\sharp}-a^{\sharp}\right) $~\eqref{in21_1} & $\textendash\textendash$\footnote{The notation ``$\textendash\textendash$'' means that whether there is an improved lower bound for this interval of the parameter is unknown, or the corresponding bound is not an improved one.}  \\
    $L^{\sharp}_{\alpha,s,d}=d^{\frac{2s(1-\alpha)-\alpha}{2s}} \left[ \frac{a^{\sharp}\sqrt{d}+\left(\|\mathcal{P^{\sharp}}(\rho)\|_{\tr}-b^{\sharp}\right)}{a^{\sharp}(\sqrt{d}+1)} \right]^{\frac{\alpha}{s}}-1$ ~\eqref{in22_2}& $\bar{L}^{\sharp}_{\alpha,s,d}=\frac{\|\mathcal{P^{\sharp}}(\rho)\|_{\tr}-b^{\sharp}}{a^{\sharp}d^\alpha}-1$~\eqref{in2}\\
    $L^{\sharp}_{q,p,d}=1 - \left[ 1 - \frac{1 - d^{\frac{(1-q)p}{p-1}}}{(a^{\sharp})^2(d - 1)^2} \left( \|\mathcal{P}^{\sharp}(\rho)\|_{\tr}-b^{\sharp}-a^{\sharp} \right)^2 \right]^{\frac{p-1}{p}}$~\eqref{in29_1_q_2} &$\bar{L}^{\sharp}_{q,p,d}=\frac{1 - d^{1-q}}{(a^{\sharp})^2(d - 1)^2} \left(\|\mathcal{P}^{\sharp}(\rho)\|_{\tr}-b^{\sharp}- a^{\sharp} \right)^2$~\eqref{in30_5} \\
    $L^{\sharp}_{q,s,2}=\frac{1 - 2^{1-q}}{(a^{\sharp})^2(2 - 2^{2-s})} \left(\|\mathcal{P}^{\sharp}(\rho)\|_{\tr}-b^{\sharp} -a^{\sharp} \right)^2$ ~\eqref{in24_6} & $\bar{L}^{\sharp}_{q,s,2}=1-\left[1- \frac{1-d^{-q}}{(a^{\sharp})^2(d-1)^2} (\|\mathcal{P}^{\sharp}(\rho)\|_{\tr}-b^{\sharp}-a^{\sharp})^2\right]^{\frac{q-1}{q}}$~\eqref{in25_7}\\	$L^{\sharp}_{q,d}=\frac{1 - d^{1-q}}{(a^{\sharp})^2(d - 1)^2} \left(\|\mathcal{P}^{\sharp}(\rho)\|_{\tr}-b^{\sharp}- a^{\sharp} \right)^2$ ~\eqref{in25_q_3}  	& $\textendash\textendash$ 			    	       	
\end{tabular}
\end{ruledtabular}
\end{table*}

As expected, by replacing the ($N$, $M$)-POVM in Eq.~\eqref{eqsz} with GSIC-POVM and SIC-POVM, respectively, we can also derive the associated lower bounds. Let $\{|e_{k}\ra\}\ (k=1, \cdots, d^2)$ be an orthonormal basis of $\mathbb{C}^{d^2}$. Let $\{P_{1}, P_2, \cdots, P_{d^{2}}\}$ be a GSIC-POVM and $\{E_1, E_2, \cdots, E_{d^2}\}$ be a SIC-POVM, we define 
\bea
\mathcal{P'}(\rho):=\sum_{k,l=1}^{d^2}\tr\left[\rho(P_{k} \ot P_{l})\right]|e_{k}\ra\la e_{l}|,\label{P'}\\
\mathcal{P''}(\rho):=\sum_{k,l=1}^{d^2}\tr\left[\rho(E_{k} \ot E_{l})\right]|e_{k}\ra\la e_{l}|.\label{P''}
\eea  
Let $a'=\frac{x' d^{3}-1}{d(d^2-1)}$, $b'=\frac{1-x' d}{d^2-1}$, and $a''=b''=\frac{1}{d(d+1)}$, where $x'=1/d^3+t'^2(d-1)(d+1)^3$ with
\beax 
\frac{1}{d^2\lambda_{\max}}\leq t'\leq \frac{1}{d^2|\lambda_{\min}|}, 
\eeax 
$\lambda_{\max}$ and $\lambda_{\min}$ are the maximal and minimal eigenvalues from all eigenvalues of $H_{\alpha, k}$s in which $\{P_{1}, P_2, \cdots, P_{d^{2}}\}$ is regarded as a $(1, d^2)$-POVM, respectively.

For simplicity, we denote by $\mP^\sharp$ either $\mathcal{P}$, or $\mathcal{P}'$, or $\mathcal{P}''$, $a^{\sharp}$, and $b^{\sharp}$ follow the same rule. The lower bounds induced from $\mP^\sharp$ are listed in Table~\ref{tab:table1}, and we will compare them in the next section. For convenience, we collectively call the lower bounds in Eqs.\eqref{in25_q_3}, \eqref{in24_6}, \eqref{in29_1_q_2}, \eqref{in21_1} and~\eqref{in22_2} (i.e., Corollary~\ref{cor1} and Corollary~\ref{cor2}) the original lower bounds (OLBs) (since they are directly derived from the original bounds associated with the PPT criterion and the realignment criterion), and the lower bounds in Eqs.~\eqref{in30_5}, \eqref{in25_7} and~\eqref{in2} (i.e., Theorem~\ref{th1}, Theorem~\ref{th2}, and Theorem~\ref{th4}) the improved lower bounds (ILBs).

\section{Examples}\label{sec4}

In this section, we present several examples to demonstrate the advantages of various lower bounds we derived: (i) the lower bound via $\mP$ can outperform those obtained via $\mP'$, and all of them can outperform the one via $\mP''$, (ii) all these measurement-induced bounds exhibit an advantage compared with the original one induced by the PPT criterion and the realignment criterion, and (iii) based on the ($N$, $M$)-POVM, the lower bounds in the our theorems (i.e., the ILBs) can be tighter than those in corollaries (i.e., the OLBs).


\begin{example}\label{eex1}                                                     
We consider the $3\ot3$ PPT entangled state given in Ref.~\cite{Bennett1999prl},
\bea\label{ex1}
\rho =\frac{1}{4}\left(I- \sum_{i=0}^{4} |\psi_i\ra\la\psi_i|\right),
\eea
where $I$ denotes the $9\times9$ identity matrix and
\beax
	|\psi_0\ra &=& \frac{|0\ra(|0\ra - |1\ra)}{\sqrt{2}}, \ 
	|\psi_1\ra = \frac{(|0\ra - |1\ra)|2\ra}{\sqrt{2}}, \\
	|\psi_2\ra &=& \frac{|2\ra(|1\ra - |2\ra)}{\sqrt{2}},\  
	|\psi_3\ra = \frac{(|1\ra - |2\ra)|0\ra}{\sqrt{2}}, \\
	|\psi_4\ra &=& \frac{(|0\ra + |1\ra + |2\ra)(|0\ra + |1\ra + |2\ra)}{3}.
\eeax
\end{example}

\begin{table*}[htb]
\caption{\label{tab:table2} Comparison of lower bounds of parameterized entanglement monotone for the PPT state in Eq.~\eqref{ex1}. }	
\begin{ruledtabular}
\renewcommand{\arraystretch}{1.3}
\begin{tabular}{cccccc}
	    &$\alpha=0.2$ & $\alpha=0.51$ & $q=1.8$ & $q=2.5$ & $q=3$\\\colrule
		($8$,$2$)-POVM  & $L_{0.2,3}=0.0677$ & $L_{0.51,0.499,3}=0.0124$        &$L_{1.8,1.\dot{3}\dot{6},3}=0.000823$ & $L_{2.5,3}=0.0019$ & $L_{3,3}
			=0.0021$ \\			
       GSIC-POVM & $L'_{0.2,3}=0.0674$ & $L'_{0.51,0.499,3}=0.0122$ &$L'_{1.8,1.\dot{3}\dot{6},3}=0.000815$ &$L'_{2.5,3}=0.0018$ & $L'_{3,3}=0.0020$\\
	   SIC-POVM	& $L''_{0.2,3}=0.0659$ & $L''_{0.51,0.499,3}=0.0114$ & $L''_{1.8,1.\dot{3}\dot{6},3}=0.000778$ & $L''_{2.5,3}=0.0018$ & $L''_{3,3}=0.0019$\\
		Eqs.\eqref{inq_2}-\eqref{eq14} & $\check{L}_{0.2,3}=0.0615$ & $\check{L}_{0.51,0.499,3}=0.0091$ & $\check{L}_{1.8,1.\dot{3}\dot{6},3}=0.000680$ & $\check{L}_{2.5,3}=0.0015$ & $\check{L}_{3,3}=0.0017$
\end{tabular}
\end{ruledtabular}
\end{table*}

By taking $\alpha=0.2$, $0.51$ and $q=1.8$, $2.5$, $3$, respectively, we get the lower bounds of parameterized entanglement monotone based on $(8,2)$-POVM (see Appendix~\ref{sec8}) with $x = 0.7506$, GSIC-POVM with $x'=0.0498$, SIC-POVM, and Eqs.~\eqref{inq_2}-\eqref{eq14}, correspondingly. We list these bounds in Table~\ref{tab:table2} from which one can clearly observe that the lower bounds of parameterized entanglement monotone based on the $(8,2)$-POVM are larger than those based on other methods: (i) $L_{0.2,3}>L'_{0.2,3}>L''_{0.2,3}>\check{L}_{0.2,3}>0$, (ii) $L_{0.51,0.499,3}>L'_{0.51,0.499,3}>L''_{0.51,0.499,3}>\check{L}_{0.51,0.499,3}>0$, $L_{1.8,1.\dot{3}\dot{6},3}>L'_{1.8,1.\dot{3}\dot{6},3}>L''_{1.8,1.\dot{3}\dot{6},3}>\check{L}_{1.8,1.\dot{3}\dot{6},3}>0$, $L_{{2.5,3}}>L'_{{2.5,3}}>L''_{2.5,3}>\check{L}_{2.5,3}>0$, and $L_{{3,3}}>L'_{{3,3}}>L''_{3,3}>\check{L}_{3,3}>0$ where $\check{L}_{*}$ denote the lower bounds in Eqs.~\eqref{inq_2}-\eqref{eq14} with respect to the corresponding parameter.

\begin{example}\label{eex2}
Consider a $3\ot3$ pure state 
\bea \label{ex2}
|\Phi\ra=\sqrt{\theta}|\phi_1\ra+\sqrt{1-\theta}|\phi_2\ra,
\eea  
where $|\phi_1\ra=\frac{1}{\sqrt{2}}(|0\ra+|1\ra)|2\ra$ and $|\phi_2\ra=\frac{1}{\sqrt{3}}(|00\ra+|11\ra+|22\ra)$.
\end{example}

\begin{figure}[htp]
	\centering
	\includegraphics[width=0.45\textwidth]{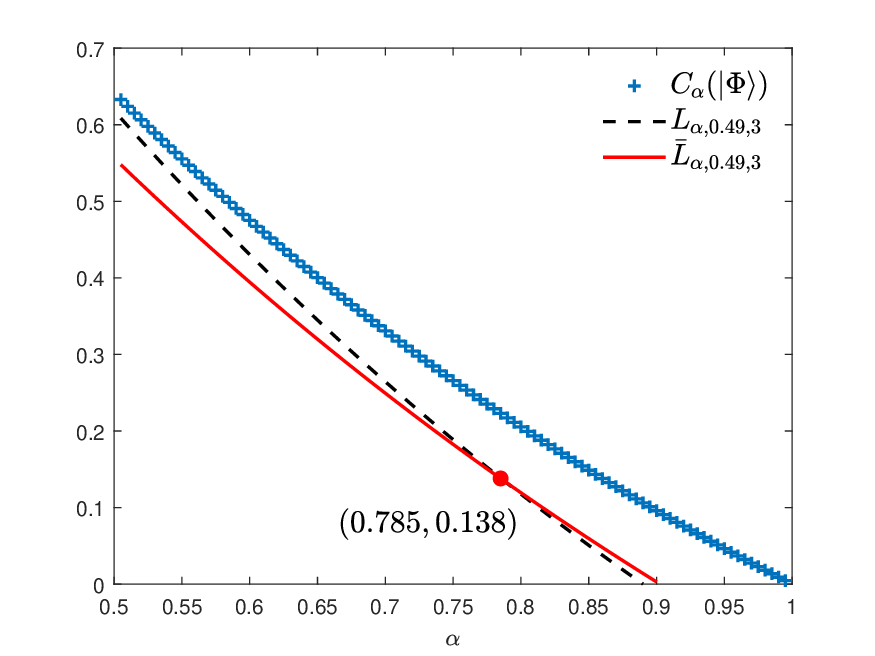}\\
	(a)\\
	\includegraphics[width=0.45\textwidth]{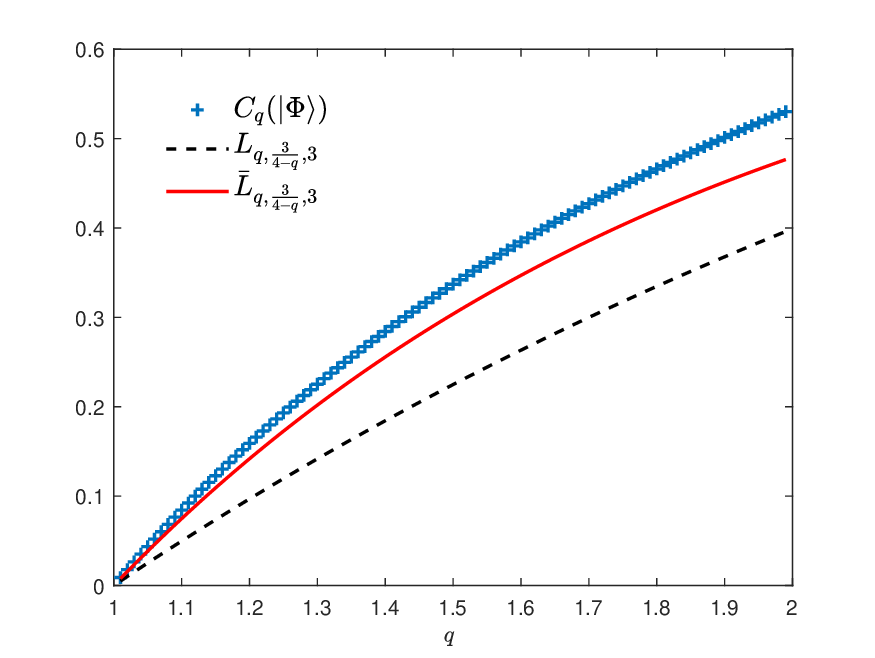}\\
	(b)
	\caption{\label{fig1} The comparison of (a) $C_{\alpha}$ and its lower bounds $L_{\alpha,0.49,3}$, $\bar{L}_{\alpha,0.49,3}$, and that of (b) $C_q$ and its lower bounds $L_{q,\frac{3}{4-q},3}$, $\bar{L}_{q,\frac{3}{4-q},3}$ for $|\Phi\ra$ given in Eq.~\eqref{ex2} with $\theta=0.2$. }
\end{figure}

We take a ($8$,$2$)-POVM (see Appendix~\ref{sec8}) with $x=3/4+t^2(\sqrt{2}+1)^2$ and $t\in [-0.2536,0.2536]$. We plot in Fig.~\ref{fig1} (a) $C_\alpha(|\Phi\ra)$ and the lower bounds based on Eq.~\eqref{in22_2} and Eq.~\eqref{in2}, respectively,
and in Fig.~\ref{fig1} (b) $C_q(|\Phi\ra)$ and the lower bounds based on Eq.~\eqref{in29_1_q_2} and Eq.~\eqref{in30_5},
respectively, for the case of $\theta=0.2$, $\alpha\in (0.5,1)$ and $q\in (1,2)$. It is evident that the lower bound in Theorem~\ref{th4} outperforms that in Corollary~\ref{cor2} for $|\Phi\ra$ whenever $\alpha\in(0.785,1)$, and that the lower bound in Theorem~\ref{th1} outperforms that of Corollary~\ref{cor1} for $|\Phi\ra$ whenever $q\in(1,2)$.

\begin{example}\label{eex3}
Consider a $2\ot2$ pure state 
\bea\label{ex3}
|\Psi\ra=\sqrt{\theta}|\psi_1\ra+\sqrt{1-\theta}|\psi_2\ra,
\eea
where $|\psi_1\ra=|01\ra$, and $|\psi_2\ra=\frac{1}{\sqrt{2}}(|00\ra+|11\ra)$.	
\end{example}

Taking a ($3$,$2$)-POVM (see Appendix~\ref{sec8}) with $x = 1/2 + t^2(\sqrt{2} + 1)^2$, $t\in[-0.2929,0.2929]$, $\theta=0.2$, we plot Fig~\ref{fig2} (a): $C_\alpha(|\Psi\ra)$ ($0.5<\alpha<1$) and its lower bounds from Eq.~\eqref{in22_2} and Eq.~\eqref{in2}; Fig.~\ref{fig2} (b): $C_q(|\Psi\ra)$~$(1<q<2)$ and its lower bounds from Eq.~\eqref{in29_1_q_2} and Eq.~\eqref{in30_5}; Fig.~\ref{fig2} (c): $C_q(|\Psi\ra)$~$(2\leqslant q<3)$ and its lower bounds from Eq.~\eqref{in24_6} and Eq.~\eqref{in25_7}. One can clearly show that the lower bounds in Theorems~\ref{th1}, \ref{th2} are more tighter than that in Corollaries~\ref{cor1} (see Fig.~\ref{fig2}).

\begin{example}\label{eex4}
Consider the isotropic states
\bea\label{ex4}
\rho_F=\frac{1-F}{d^{2}-1}\left(I-\left|\Phi^{+}\right>\left<\Phi^{+}\right|\right)+F\left|\Phi^{+}\right>\left<\Phi^{+}\right|,
\eea
where $|\Phi^{+}\ra=\frac{1}{\sqrt{d}} \sum_{i=1}^{d}|i i\ra$ is a maximally entangled pure state, $F$ is the fidelity of $\rho_F$ with respect to $|\Phi^{+}\ra$, $F=\la\Phi^{+}|\rho_F|\Phi^{+}\ra$, and $0\leqslant F\leqslant 1$.
\end{example}

Following~\cite{Terhal2000prl,Rungta2003pra,Lee2003pra,Wang2016pra} and based on the arguments in Ref.~\cite{Yang2021pra,Wei2022jpamt}, we can calculate the parameterized entanglement monotone for the isotropic state for $1/2<\alpha<1$ and $1< q<2$ as shown below (the proof is given in Appendix~\ref{sec7}).

\begin{figure}[htp]
	\centering
	\includegraphics[width=0.42\textwidth]{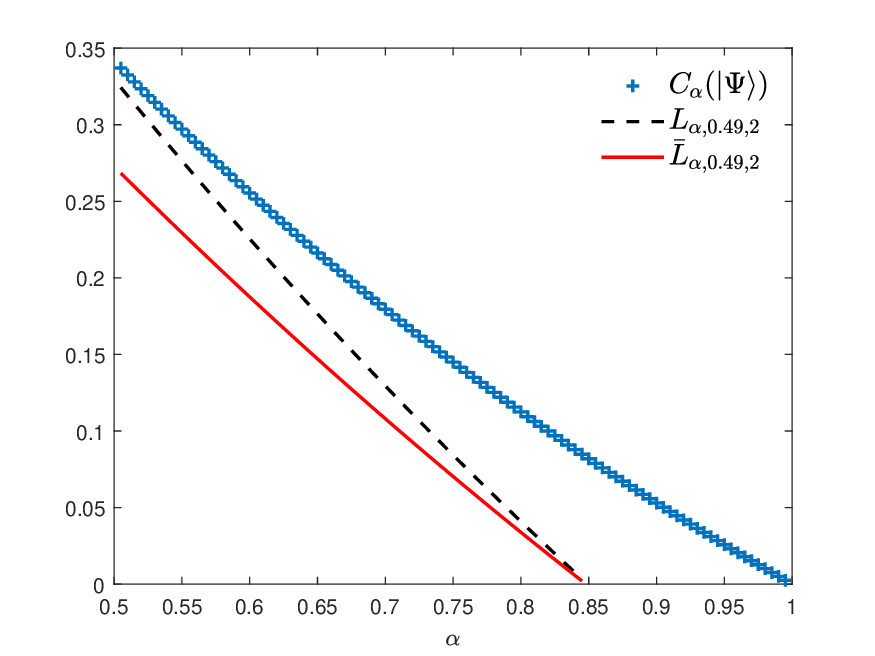}\\
	(a)\\
	\includegraphics[width=0.42\textwidth]{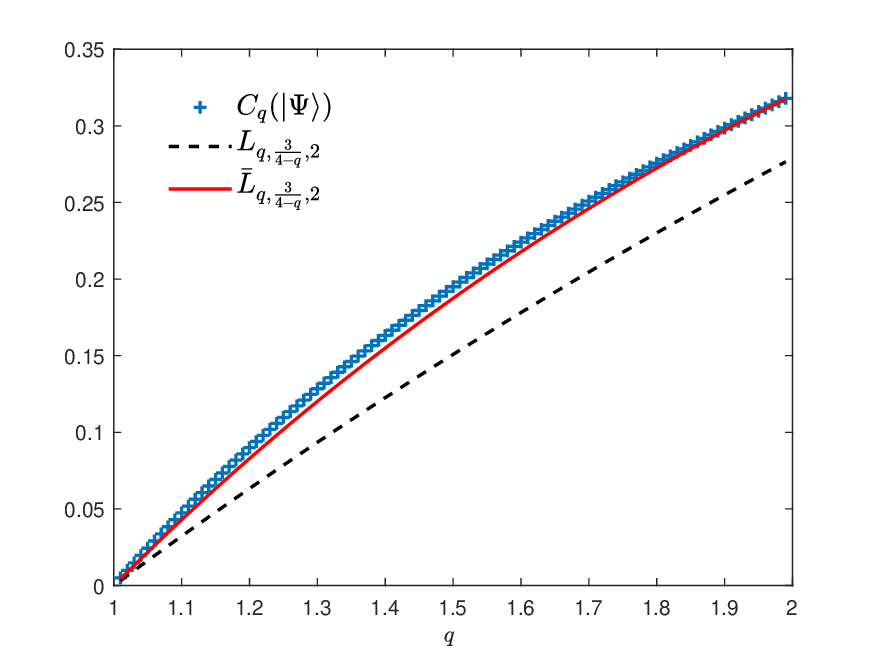}\\
	(b)\\
	\includegraphics[width=0.42\textwidth]{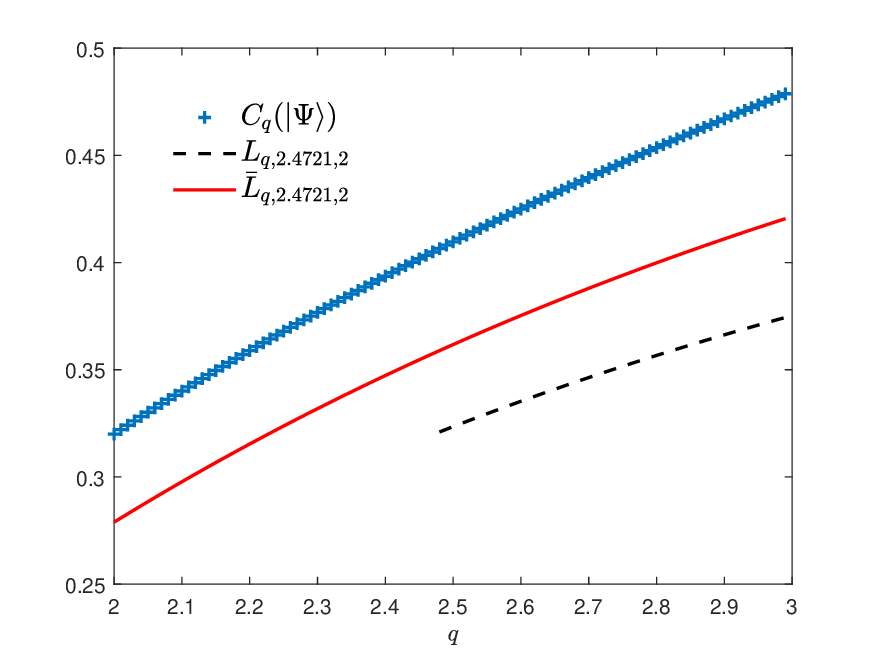}\\
	(c)
	\caption{\label{fig2}The comparison of the different lower bounds for $|\Psi\ra$ given in Eq.~\eqref{ex3} with $\theta=0.2$: (a) $C_{\alpha}$, $L_{\alpha,0.49,2}$, and $\bar{L}_{\alpha,0.49,2}$; (b) $C_q$, $L_{q,\frac{3}{4-q},2}$, and $\bar{L}_{q,\frac{3}{4-q},2}$; (c) $C_q$, $L_{q,2.4721,2}$, and $\bar{L}_{q,2.4721,2}$.}
\end{figure}

\smallskip 
\noindent{\bf Lemma}~~{\textit{For $1/2<\alpha<1$, 
		\bea 
		C_\alpha(\rho_F)=\text{\rm co}\left[ \xi_{\alpha}(\rho_F)\right] 
		\eea 		
		with
		\bex
		\xi_{\alpha}(\rho_F)=\gamma^{2\alpha}+(d-1) \delta^{2\alpha}-1,
		\eex
		and for $1<q<2$,
		\bea 
		C_q(\rho_F)=\text{\rm co}\left[ \xi_{q}(\rho_F)\right] 
		\eea 		
		with
		\bex
		\xi_{q}(\rho_F)=1-\gamma^{2q}-(d-1)\delta^{2q},
		\eex
		where $\gamma=\frac{1}{\sqrt{d}} \left[ \sqrt{F} + \sqrt{(d-1)(1-F)} \right]$, $\delta=\frac{1}{\sqrt{d}} \left( \sqrt{F} - \frac{\sqrt{1-F}}{\sqrt{d-1}} \right)$, $F>\frac{1}{d}$, and ${\rm co}(f)$ denotes the convex hull of the function $f$ (i.e., ${\rm co}(f)$ is the largest convex function that is upper bounded by the function $f$).}}

\smallskip

\begin{figure}[htp]
	\centering
	\includegraphics[width=0.48\textwidth]{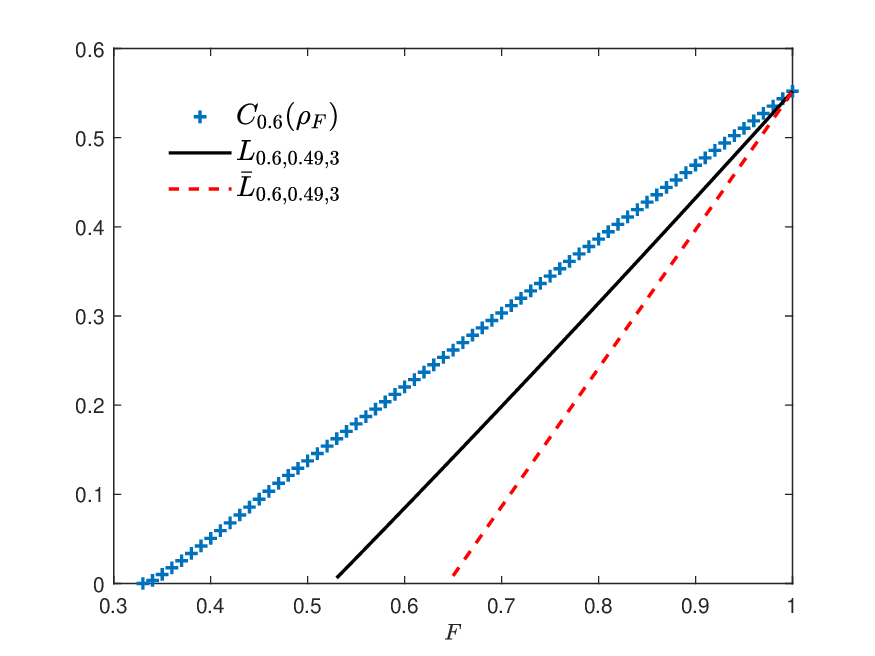}\\
	(a)\\
	\includegraphics[width=0.48\textwidth]{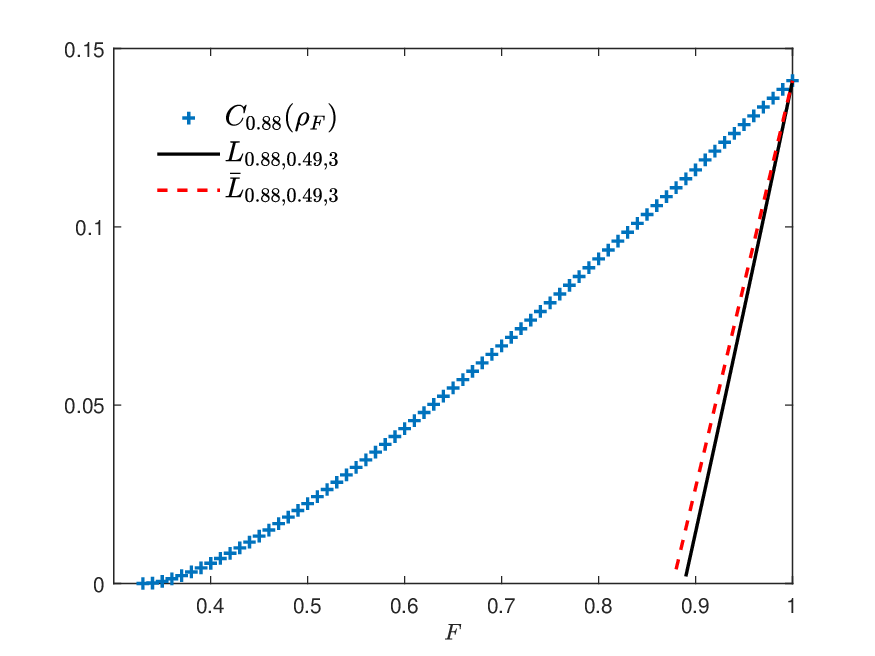}\\
	(b)
	\caption{\label{fig3}The comparison of $C_{\alpha}$ and its lower bounds $L_{\alpha,0.49,3}$, $\bar{L}_{\alpha,0.49,3}$ for the $3\ot 3$ isotropic state $\rho_F$ with (a) $\alpha=0.6$ and (b) $\alpha=0.88$, respectively.}
\end{figure}

By the Lemma above, when $d=3$, we can calculate immediately that
\bex\label{eq0.6}
C_{0.6}\left(\rho_F\right)=
\begin{cases}
	0,                 & F\leqslant 1/3,\\
	\xi_{0.6}(\rho_F), & 1/3<F\leqslant 0.481,\\
	0.829F-0.277,      & 0.481<F\leqslant 1,
\end{cases}
\eex
\bex\label{eq0.88}
C_{0.88}\left(\rho_F\right)=
\begin{cases}
	0,                  & F \leqslant 1/3, \\
	\xi_{0.88}(\rho_F), & 1/3< F \leqslant 0.907,\\
    0.247F-0.106,       & 0.907 < F \leqslant 1,
\end{cases}
\eex
where 
\beax
\xi_{\alpha}(\rho_F)=\left(\frac{2-F+\vartheta_1}{3}\right)^{\alpha}+2\left(\frac{1+F-\vartheta_1}{6}\right)^{\alpha}-1 
\eeax
with $\vartheta_1=2\sqrt{2F(1-F)}$, and that for the case of $1<q<2$, 
\bex
C_{1.5}\left(\rho_F\right)=
\begin{cases}
	 0,               & F \leqslant 1/2, \\
	\xi_{1.5}(\rho_F),& F > 1/2
\end{cases}
\eex
whenever $d = 2$, where 
\bex
\xi_{q}(\rho_F)=1-\left(\frac{1+\vartheta_2}{2}\right)^q+\left(\frac{1-\vartheta_2}{2}\right)^q
\eex 
with $\vartheta_2 =2\sqrt{F(1-F)}$, 
\bex
C_{1.5}\left(\rho_F\right)=
\begin{cases}
	0,                 & F \leqslant 1/3, \\
	\xi_{1.5}(\rho_F), & 1/3 < F \leqslant 0.949,\\
	0.922F-0.499,      & 0.949 < F \leqslant 1
\end{cases}
\eex
with
\beax
\xi_{q}(\rho_F)=1-\left(\frac{2-F+\vartheta_1}{3}\right)^q-2\left(\frac{1+F-\vartheta_1}{6} \right)^q
\eeax
whenever $d=3$,
\bex
C_{1.5}\left(\rho_F\right)=
\begin{cases}
	0,                 & F \leqslant 0.25, \\
	\xi_{1.5}(\rho_F), & 0.25 < F \leqslant 0.881,\\
	0.882F-0.382,      & 0.881 < F \leqslant 1
\end{cases}
\eex
with
\beax
\xi_{q}(\rho_F)=1-\left(\frac{3-2F+\vartheta_3}{4}\right)^q-3\left(\frac{1+2F-\vartheta_3}{12} \right)^q
\eeax
whenever $d=4$,	where $\vartheta_3=2\sqrt{3F(1-F)}$.

\begin{figure}[htp]
	\centering			
	\includegraphics[width=0.45\textwidth]{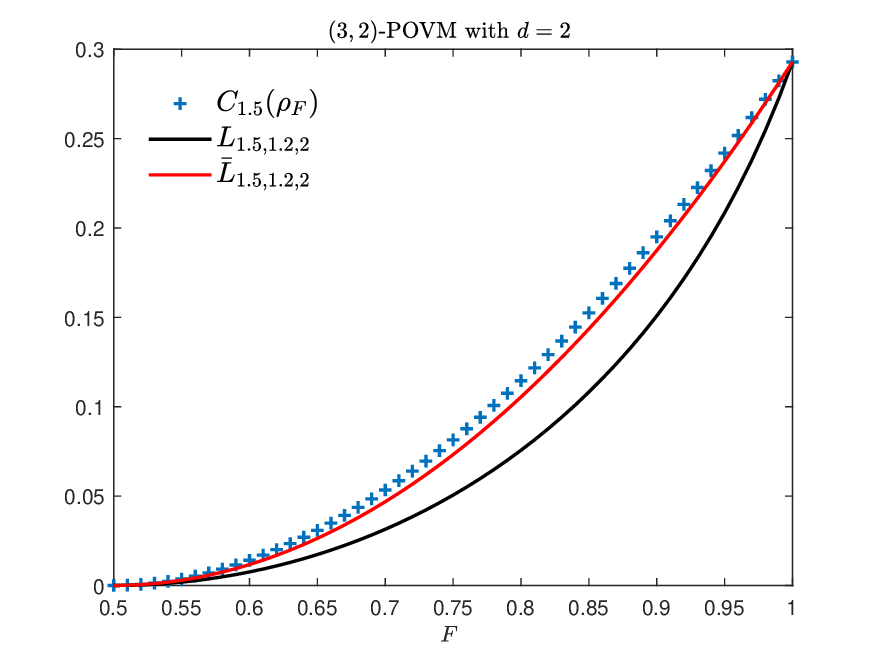}\\		
	\includegraphics[width=0.45\textwidth]{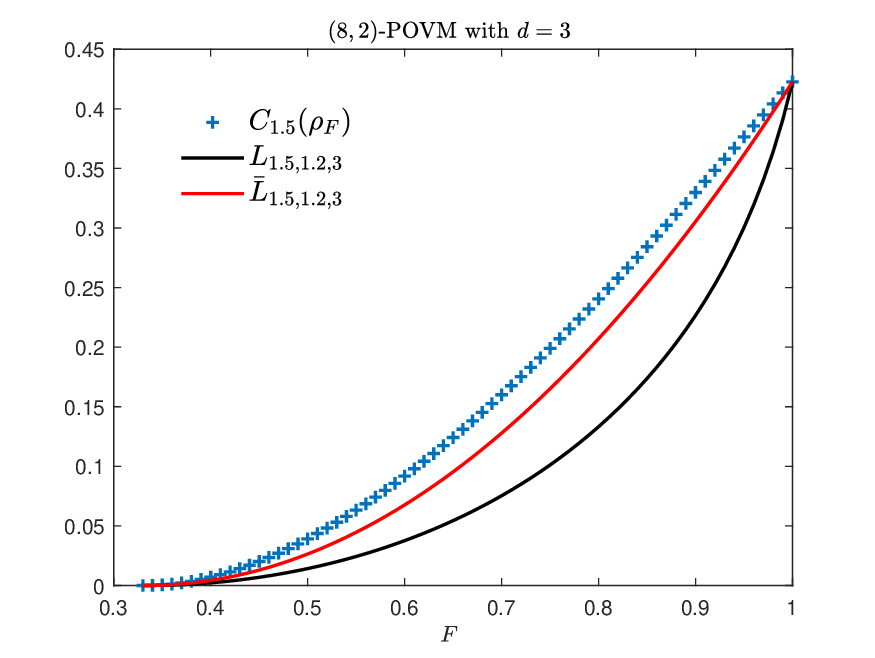}\\
	\includegraphics[width=0.45\textwidth]{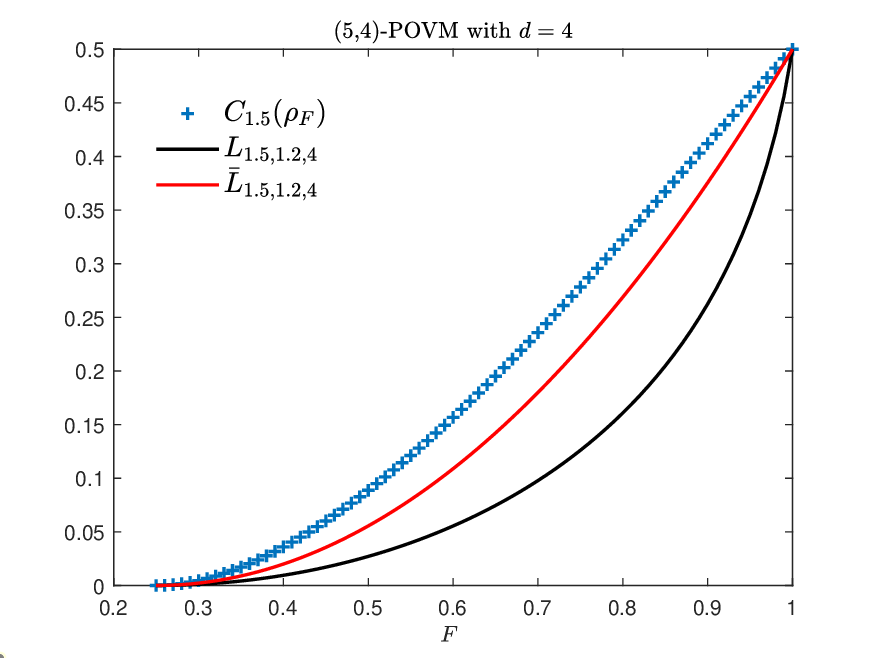}\\
	\caption{\label{fig4}The comparison between $C_{1.5}(\rho_F)$ and its lower bounds for the cases $d=2$, $3$, $4$, respectively.}
\end{figure}

In order to compare the lower bounds of $\rho_F$, we adopt the ($3$,$2$)-POVM (see Appendix~\ref{sec8}) with $x = 1/2 + t^2(\sqrt{2} + 1)^2$, $t\in[-0.2929,0.2929]$, when $d=2$, ($8$,$2$)-POVM (see Appendix~\ref{sec8}) with $x=3/4+t^2(\sqrt{2}+1)^2$ and $t\in [-0.2536,0.2536]$ for the case of $d=3$, and employ the ($5$,$4$)-POVM (see Appendix~\ref{sec8}) with $x=1/4+27t^2$ and $t\in [-0.0572, 0.0680]$ for the $4\ot 4$ system. Fig.~\ref{fig3}, which plots $C_{0.6}(\rho_F)$ and $C_{0.88}(\rho_F)$ alongside their corresponding lower bounds given by Eq.~\eqref{in22_2} and Eq.~\eqref{in2} respectively, shows that for $d=3$ and $F\in(1/3,1]$, $\bar{L}_{0.88,0.49,3}$ outperforms $L_{0.88,0.49,3}$, whereas $L_{0.6,0.49,3}$ outperforms $\bar{L}_{0.6,0.49,3}$ for such states. Fig.~\ref{fig4}, which shows a comparison between $C_{1.5}$ and its corresponding lower bounds given by Eqs.~\eqref{in29_1_q_2} and~\eqref{in30_5} for $d = 2, 3, 4$ and $F\in(1/d,1]$, clearly demonstrates that $\bar{L}_{1.5,1.2,d}$ outperforms $L_{1.5,1.2,d}$ at $q=1.5$ for all the considered values of $d$. It is worth mentioning that for any $q\in(1,2)$, we can determine the convex interval of the function $\xi_q(\rho_F)$ and compare the tightness of $L_{q,p,d}$ and $\bar{L}_{q,p,d}$.

\section{Conclusion}\label{sec5}

Here we have presented several lower bounds for the parameterized entanglement monotone which is a class of a wide-ranged entanglement measures in the light of the informationally complete ($N, M$)-POVM. These lower bounds show advantages over the previous ones in literature. We also find that the GSIC-POVM and the SIC-POVM, as special cases of ($N, M$)-POVM, are not as good as ($N, M$)-POVM in general when they are used in evaluating entanglement. The original method associated with the PPT criterion and the realignment criterion seems even worse for this task. Based on the ($N$, $M$)-POVM, we derive the lower bounds for two-qudit states in the regimes $\frac12<\alpha<1$ and $1<q<2$, as well as for two-qubit states in the regime $2\leqslant q<3$, and demonstrate via explicit examples that these bounds can be tighter than those available in literature. We also derive an analytical expression for the parameterized entanglement monotone for the isotropic state in the regimes $\frac12<\alpha<1$ and $1<q<2$, which is an absence in the literature along this line previously. Since the measurements are practical and feasible in laboratory settings, the lower bounds we proposed are computable and could be applied in quantum information processing.

\begin{acknowledgements}
Y.G is supported by the National Natural Science Foundation of China under Grant Nos.~12471434 and 11971277, the Program for Young Talents of Science and Technology in Universities of Inner Mongolia Autonomous Region under Grant No. NJYT25010, and the High-Level Talent Research Start-up Fund of Inner Mongolia University under Grant No. 10000-A23207007. S. D is supported by the National Natural Science Foundation of China under Grant No.~12271452.
\end{acknowledgements}

\appendix

\section{The proof of Lemma}\label{sec7}

To simplify the calculation of convex-roof extended entanglement measures for isotropic states, we adopt the conclusion in Refs.~\cite{Lee2003pra,Wei2022jpamt,Terhal2000prl}. Let $E$ be a convex-roof extended quantum entanglement measure. Let $\mathcal{T}_{\text{iso}}$ be the $(U \otimes U^{*})$-twirling operator which is defined by~\cite{Wei2022jpamt,Rungta2003pra,Lee2003pra,Terhal2000prl}
\bex
\mathcal{T}_{\text{iso}}(\rho)=\int dU \,(U\otimes U^*)\rho(U\otimes U^*)^\dagger
\eex
under some given orthonormal basis of $\mH^{A,B}$ with the standard (normalized) Haar measure $dU$ on the group consisting of all unitary operators acting on $\mH^{A,B}$. This twirling operator reduces any two-qudit state $\rho$ to an isotropic state, i.e., $\mathcal{T}_{\text{iso}}(\rho)=\rho_{F(\rho)}$, where $F(\rho) = \la\Phi^+|\rho|\Phi^+\ra$ is the fidelity of $\rho$ with respect to $|\Phi^{+}\ra$. Then the function $\xi$ on $\mathcal{T}_{\text{iso}}(\mS^{AB})$ defined by
\be\label{eqxi}
\xi(\rho_F)=\min_{|\psi\ra}\left\{E(|\psi\ra):\,|\psi\ra\in \mH^{AB},\  \mathcal{T}_{\text{iso}}(|\psi\ra\la\psi|)=\rho_F\right\}
\ee
satisfies~\cite{Lee2003pra}
\be
E(\rho_F)=\mathrm{co}[\xi(\rho_F)],
\ee
where ${\rm co}(\xi)$ denotes the convex hull of $\xi$.

As in Refs.~\cite{Lee2003pra,Wei2022jpamt,Terhal2000prl}, we have
\bex
\mathcal{T}_{\text{iso}}(|\psi\ra\la\psi|)=\rho_{F(|\psi\ra\la\psi|)}=\rho_{F(\vec{\lambda},V)},
\eex
where $\vec{\lambda}=(\lambda_1,\lambda_2,\cdots,\lambda_d)$ is the Schmidt coefficient sequence of $|\psi\ra$, $|\psi\ra=\sum_{i=1}^{d}\lambda_i\,U_A\ot U_B|ii\ra$ for some unitary operators $U_{A,B}$, and 
\beax
F(\vec{\lambda}, V)=|\la\Phi^+|\psi\ra|^2&=&\frac{1}{d}\left|\sum_{i,j=1}^{d}\lambda_i\la j|U_A|i\ra\la j|U_B|i\ra\right|^2\\
&=&\frac{1}{d}\left|\sum_{i=1}^{d}\lambda_i(U_A^TU_B)_{ii}\right|^2\\
&=& \frac{1}{d}\left|\sum_{i=1}^{d}\lambda_i V_{ii}\right|^2
\eeax
with $V_{ij}=\la i|V|j\ra$ and $V=U_A^T U_B$. Then the function $\xi_{q}(\rho_F)$ as defined in Eq.~\eqref{eqxi} takes the form
\bex
\xi_{q}(\rho_F)=\min_{\{\vec{\lambda}, V\}} \left\{ C_{q}(\vec{\lambda}) :\,\frac{1}{d}\left|\sum_{i=1}^{d} \lambda_i V_{ii}\right|^2=F\right\}.
\eex
To investigate the impact of $V$ on $\xi_q(\rho_F)$, for any given unitary operator $V$, we consider the function
\bex
\mC_{q,V}(F)=\min_{\vec{\lambda}}\left\{C_{q}(\vec{\lambda}):\,\frac{1}{d}\left|\sum_{i=1}^{d}\lambda_i V_{ii} \right|^2=F\right\}.
\eex
Let $\vec{\varsigma}$ be the Schmidt coefficient sequence that reaches the minimum value of $\mC_{q,V}(F)$. It can be easily shown that $F(\vec{\varsigma}, V)\leqslant F(\vec{\varsigma}, I)$. Thus $\mC_{q,V}(F)$ attains its minimum value at $V=I$ (i.e., $\mC_{q,V}(F)\geqslant \mC_{q,I}(F)$).
Therefore, we have
\be\label{eqB}
\xi_{q}(\rho_F)=\mC_{q,I}(F)=\min_{\vec{\lambda}}\left\{C_{q}(\vec{\lambda}):\,\frac{1}{d}\left|\sum_{i=1}^{d} \lambda_i\right|^2=F\right\}.
\ee
It is worth noting here that, replacing $q$ with $\alpha$, for $1/2<\alpha<1$, we can get 
\be\label{eqC}
\xi_{\alpha}(\rho_F)=\min_{\vec{\lambda}}\left\{C_{\alpha}(\vec{\lambda}):\,\frac{1}{d}\left|\sum_{i=1}^{d} \lambda_i\right|^2=F\right\}.
\ee

Following the method in~\cite{Rungta2003pra,Yang2021pra,Wei2022jpamt}, using the Lagrange multipliers, one can minimize \eqref{eqB} and \eqref{eqC} subject to the constraints
\bex
\sum_{i=1}^{d} \lambda_i^2 = 1, \quad \sum_{i=1}^{d} \lambda_i = \sqrt{Fd}
\eex
with $\lambda_i \geqslant 0$ and $Fd \geqslant 1$. For $1<q<2$, the condition for an extremum can be expressed as
\bea\label{equation1}
\lambda_i^{2q-1}+\mu_1\lambda_i+\mu_2=0,
\eea
where $\mu_1$ and $\mu_2$ denote the Lagrange multipliers.
Similarly, for $1/2<\alpha<1$, we obtain
\bea\label{equation2}
\lambda_i^{2\alpha-1}+\mu_1\lambda_i+\mu_2=0.
\eea
It is clear that $f(\lambda_i) = \lambda_i^{2q-1}$ is convex for $1 < q < 2$ and $f(\lambda_i) = \lambda_i^{2\alpha-1}$ is concave for $1/2<\alpha<1$. Since both convex and concave functions can intersect a linear function at most two points, \eqref{equation1} and~\eqref{equation2} have at most two nonzero solutions respectively. Let $\gamma$ and $\delta$ denote these two possibly positive solutions~($\gamma>\delta$), one has
\beax
\lambda_j =
\begin{cases}
	    \gamma, & j = 1, \dots, n, \\
		\delta, & j = n + 1, \dots, n + m, \\
		0, & j = n + m + 1, \dots, d,
\end{cases}
\eeax
where $n + m \leq d$ and $n \geq 1$. The minimization problems of~\eqref{eqB}~and~\eqref{eqC}~are thereby reduced to
\bex
\xi_{q}(\rho_F) = \min_{n,m} C_{q}(F),\quad
\xi_{\alpha}(\rho_F)=\min_{n,m}C_{\alpha}(F)
\eex
subject to
\bea\label{eqB1}
n\gamma^2+m\delta^2=1, \quad n\gamma+m\delta=\sqrt{Fd},
\eea
where $C_q(F)=1-n\gamma^{2q}-m\delta^{2q}$ and $C_{\alpha}(F)=n\gamma^{2\alpha}+m\delta^{2\alpha}-1$. The solutions of Eq.~\eqref{eqB1} are
\beax
\gamma_{\text{nm}}^\pm(F)=\frac{n\sqrt{Fd}\pm\sqrt{nm(n+m-Fd)}}{n(n+m)},
\eeax
\beax
\delta_{\text{nm}}^\pm(F)&=&\frac{\sqrt{Fd}-n\gamma_{\text{nm}}^\pm}{m}\\
&=&\frac{m\sqrt{Fd}\mp \sqrt{nm(n+m-Fd)}}{m(n+m)}.
\eeax
Due to the relation $\gamma_{\text{nm}}^{\pm}(F)=\delta_{\text{nm}}^{\mp}(F)$, it is sufficient to restrict our attention to the solutions $\gamma_{\text{nm}}(F):=\gamma_{\text{nm}}^{+}(F)$ and $\delta_{\text{nm}}(F):=\delta_{\text{nm}}^{+}(F)$. These quantities are real whenever $Fd \leqslant n+m$. In addition, since $\delta_{\text{nm}}(F)$ must be non-negative, one requires $Fd\geqslant n$. Consequently, one obtains, $\delta_{\text{nm}}(F)\leqslant\frac{\sqrt{Fd}}{n + m}\leqslant\gamma_{\text{nm}}(F)$, which is consistent with the assumption $\gamma>\delta$. Here $n\geqslant1$, since the case $n=0$ would be ill-defined.

Our goal is to find the minimum values of $C_{q}(F)$ and $C_{\alpha}(F)$ across all choices of $n$ and $m$, which can be achieved by minimizing $C_{q}(F)$ and $C_{\alpha}(F)$ within the parallelogram bounded by the inequalities $1\leqslant n\leqslant Fd$ and $Fd\leqslant n+m\leqslant d$. Note that the parallelogram collapses to a line when $Fd=1$, i.e., the separable boundary. Moreover, $\gamma_{\text{nm}}=\delta_{\text{nm}}$ if and only if $n+m=Fd$, while $\delta_{\text{nm}}=0$ if and only if $n=Fd$. The derivatives of $\gamma_{\text{nm}}$ and $\delta_{\text{nm}}$ with respect to $n$ and $m$ are given by
\beax
\frac{\partial\gamma}{\partial n}&=&\frac{1}{2n}\cdot\frac{2\gamma\delta-\gamma^2}{\gamma-\delta},\ \ \ \ \
\frac{\partial\delta}{\partial n}=-\frac{1}{2m}\cdot\frac{\gamma^2}{\gamma-\delta}, 
\eeax
\beax
\frac{\partial\delta}{\partial m}&=&-\frac{1}{2m}\cdot\frac{2\gamma\delta-\delta^2}{\gamma-\delta}, \ \ 
\frac{\partial\gamma}{\partial m}=\frac{1}{2n}\cdot\frac{\delta^2}{\gamma-\delta}.
\eeax
Hence, we have the partial derivatives of $C_{q}(F)$ and $C_{\alpha}(F)$ with respect to $n$ and $m$ as follows	
\beax
\frac{\partial C_{q}}{\partial n}&=&(q-1)\gamma^{2q}-q\gamma^2\delta \cdot \frac{\gamma^{2q-2}-\delta^{2q-2}}{\gamma-\delta},	
\eeax
\beax
\frac{\partial C_{q}}{\partial m}&=&(q-1)\delta^{2q}-q\delta^2\gamma\cdot\frac{\gamma^{2q-2}-\delta^{2q-2}}{\gamma-\delta},
\eeax
\beax
\frac{\partial C_{\alpha}}{\partial n}&=&(1-\alpha)\gamma^{2\alpha}+\alpha\gamma^2\delta\cdot \frac{\gamma^{2\alpha-2}-\delta^{2\alpha-2}}{\gamma-\delta}, 
\eeax
\beax	
\frac{\partial C_{\alpha}}{\partial m}&=&(1-\alpha)\delta^{2\alpha}+\alpha\delta^2\gamma\cdot \frac{\gamma^{2\alpha-2}-\delta^{2\alpha-2}}{\gamma-\delta}.
\eeax

Let us first discuss the case $1<q<2$. To determine the sign of $\frac{\partial C_{q}}{\partial m} $, we can reduce the problem to determining that of
\beax
(q-1)(\gamma-\delta)\delta^{2q}-q\delta^2\gamma(\gamma^{2q-2}-\delta^{2q-2}),
\eeax
which is in turn equivalent to establishing the sign of 
\beax
\frac{(2q-1)\frac{\gamma}{\delta}-(q-1)}{q}-\left(\frac{\gamma}{\delta}\right)^{2q-1}.
\eeax
Setting $x=\frac{\gamma}{\delta}$ with $x\geqslant 1$, we denote 
\beax
f(x)=\frac{(2q-1)x-(q-1)}{q}-x^{2q-1}.
\eeax
A simple calculation gives $f(1)=0$ and $\frac{\partial f}{\partial x}\leqslant 0$, which implies
\beax
\frac{\partial C_{q}}{\partial m}\leqslant 0.
\eeax
For the case of $1/2<\alpha<1$, using similar arguments, we can obtain 
\beax 
\frac{\partial C_{\alpha}}{\partial m}\leqslant 0.
\eeax

Now we introduce two parameters $u=m-n$ and $v=m+n$, which respectively correspond to the motion parallel and perpendicular to the $m+n=$ constant boundaries of the parallelogram. The derivative of $C_{q}(F)$ with respect to $u$ is
\beax
\frac{\partial C_{q}}{\partial u}&=&\frac{\partial C_{q}}{\partial n}\frac{\partial n}{\partial u}+\frac{\partial C_{q}}{\partial m}\frac{\partial m}{\partial u}  \\
&=&\frac{1}{2}(q-1)\left(\delta^{2q}-\gamma^{2q}\right)+\frac{1}{2}q\gamma\delta\left(\gamma^{2 q-2}-\delta^{2 q-2}\right).
\eeax
To determine the sign of $\frac{\partial C_q}{\partial u}$, it suffices to determine the sign of 
\beax
(q-1)\left(\delta^{2 q}-\gamma^{2 q}\right) +q\gamma\delta\left(\gamma^{2 q-2}-\delta^{2 q-2}\right),
\eeax
which is further equivalent to determining the sign of 
\bex
\frac{q}{q-1}-\frac{\left(\frac{\gamma}{\delta}\right)^{2q}-1}{\left(\frac{\gamma}{\delta}\right)^{2q-1}-\frac{\gamma}{\delta}}.
\eex
Let $x = \frac{\gamma}{\delta}$ with $x> 1$. We write 
\bex
g(x)=\frac{q}{q-1}-\frac{x^{2q}-1}{x^{2q-1}-x} \ \ \text{with}\ \  \lim_{x\to1^+}g(x)=0.
\eex
Taking the derivative of $g(x)$ gives 
\bex
\frac{\partial g}{\partial x}=-\frac{x^{4q-2}-(2q-1)x^{2q}+(2q-1)x^{2q-2}-1}{(x^{2q-1}-x)^2}.
\eex
Taking $h(x)=x^{4q-2}-(2q-1)x^{2q}+(2q-1)x^{2q-2}-1$, we get
\beax
\frac{\partial h}{\partial x}&=&(4q-2)x^{4q-3}-2q(2q-1)x^{2q-1}\\
&&+(2q-1)(2q-2)x^{2q-3}\\
&=&(4q-2)x^{2q-3}[x^{2q}-qx^2+(q-1)].
\eeax
Let $l(x)=x^{2q}-qx^2+(q-1)$. Then 
\bex
\frac{\partial l}{\partial x}=2qx^{2q-1}-2qx\geqslant 0.
\eex
Since $l(1)=0$, we have $\frac{\partial h}{\partial x}\geqslant 0$ and $h(1)=0$. It follows that $\frac{\partial g}{\partial x}\leqslant 0$, which yields 
\bex
\frac{\partial C_{q}}{\partial u}\leqslant 0.
\eex
We now address the case where $1/2<\alpha<1$. Observing that
\beax
\frac{\partial C_{\alpha}}{\partial u}&=&\frac{\partial C_{\alpha}}{\partial n}\frac{\partial n}{\partial u}+\frac{\partial C_{\alpha}}{\partial m}\frac{\partial m}{\partial u}  \\
&=&\frac12(\alpha-1)\left(\gamma^{2\alpha}-\delta^{2\alpha}\right)-\frac{1}{2}\alpha\gamma\delta\left(\gamma^{2 \alpha-2}-\delta^{2\alpha-2}\right).
\eeax
An identical method yields
\bex
\frac{\partial C_{\alpha}}{\partial u}\leqslant 0.
\eex

For the case where $1<q<2$, based on the conclusion that $\frac{\partial C_{q}}{\partial m}\leqslant 0$ and $\frac{\partial C_{q}}{\partial u}\leqslant 0$, the minimum of $C_{q}(F)$ is obtained at the vertex of $n=1$ and $m=d-1$. In this way, we derive an analytical expression of the function $\xi_{q}(\rho_F)$ as follows 	
\bex
\xi_{q}(\rho_F)=1-\gamma^{2q}-(d-1)\delta^{2q},
\eex
where
\bex
\gamma=\frac{1}{\sqrt{d}}\left[\sqrt{F}+\sqrt{(d-1)(1-F)}\right],
\eex
\bex
\delta=\frac{1}{\sqrt{d}}\left(\sqrt{F}-\frac{\sqrt{1-F}}{\sqrt{d-1}}\right).
\eex
Similarly, for $1/2<\alpha<1$, we have
\bex
\xi_{\alpha}(\rho_F)=\gamma^{2\alpha}+(d-1)\delta^{2\alpha}-1.
\eex

\section{($N$, $M$)-POVM}\label{sec8}

In Section~\ref{sec4}, the Hermitian basis operator $G_{\alpha, k}$ of the $(3,2)$-POVM, under some given orthonormal basis, is given by the Pauli matrices	
\beax
G_{11}=\frac{1}{\sqrt{2}} \begin{pmatrix}0 & 1 \\ 1 & 0 \end{pmatrix}, \quad
G_{21}=\frac{1}{\sqrt{2}} \begin{pmatrix}0 & -\rmi \\ \rmi & 0 \end{pmatrix},
\eeax	
\beax
G_{31}=\frac{1}{\sqrt{2}} \begin{pmatrix}1 & 0 \\ 0 & -1 \end{pmatrix}.
\eeax

The Hermitian basis operator $G_{\alpha, k}$ of the $(8,2)$-POVM in Section~\ref{sec4}, under some given orthonormal basis, is given by the Gell-Mann matrices
\beax
G_{11}=\frac{1}{\sqrt{2}}\begin{pmatrix} 0 & 1 & 0 \\ 1 & 0 & 0 \\ 0 & 0 & 0 \end{pmatrix},\quad  G_{21}=\frac{1}{\sqrt{2}}\begin{pmatrix} 0 & -\rmi & 0 \\ \rmi & 0 & 0 \\ 0 & 0 & 0 \end{pmatrix},
\eeax
\beax
G_{31}=\frac{1}{\sqrt{2}}\begin{pmatrix} 0 & 0 & 1 \\ 0 & 0 & 0 \\ 1 & 0 & 0 \end{pmatrix},\quad  G_{41}=\frac{1}{\sqrt{2}}\begin{pmatrix} 0 & 0 & -\rmi \\ 0 & 0 & 0 \\ \rmi & 0 & 0 \end{pmatrix},
\eeax	
\beax
G_{51}=\frac{1}{\sqrt{2}}\begin{pmatrix} 0 & 0 & 0 \\ 0 & 0 & 1 \\ 0 & 1 & 0 \end{pmatrix}, \quad  G_{61}=\frac{1}{\sqrt{2}}\begin{pmatrix} 0 & 0 & 0 \\ 0 & 0 & -\rmi \\ 0 & \rmi & 0 \end{pmatrix},
\eeax	

\beax
G_{71}=\frac{1}{\sqrt{2}}\begin{pmatrix} 1 & 0 & 0 \\ 0 & -1 & 0 \\ 0 & 0 & 0 \end{pmatrix}, \quad  G_{81}=\frac{1}{\sqrt{6}}\begin{pmatrix} 1 & 0 & 0 \\ 0 & 1 & 0 \\ 0 & 0 & -2 \end{pmatrix}.
\eeax

The Hermitian basis operator $G_{\alpha, k}$ of the $(5,4)$-POVM in Section~\ref{sec4}, under some given orthonormal basis, is given by the general Gell-Mann matrices	
\[
G_{11}=\frac{1}{\sqrt{2}}\begin{pmatrix}0 & -\rmi & 0 & 0\\ \rmi & 0 & 0 & 0 \\0 & 0 & 0 & 0 \\0 & 0 & 0 & 0 \end{pmatrix}, \quad 
G_{12}=\frac{1}{\sqrt{2}}\begin{pmatrix}0 & 0 & -\rmi & 0\\ 0 & 0 & 0 & 0\\ \rmi & 0 & 0 & 0\\ 0 & 0 & 0 & 0 \end{pmatrix}, 
\]
\[
G_{13}=\frac{1}{\sqrt{2}}\begin{pmatrix}0 & 0 & 0 & -\rmi \\0 & 0 & 0 & 0 \\0 & 0 & 0 & 0 \\\rmi & 0 & 0 & 0
\end{pmatrix},	\quad 
G_{21}=\frac{1}{\sqrt{2}}\begin{pmatrix}0 & 1 & 0 & 0 \\1 & 0 & 0 & 0 \\0 & 0 & 0 & 0 \\0 & 0 & 0 & 0
\end{pmatrix},
\]
\[
G_{22}=\frac{1}{\sqrt{2}}\begin{pmatrix}0 & 0 & 0 & 0 \\0 & 0 & -\rmi & 0 \\0 & \rmi & 0 & 0 \\0 & 0 & 0 & 0
\end{pmatrix}, \quad 
G_{23}=\frac{1}{\sqrt{2}}\begin{pmatrix}0 & 0 & 0 & 0 \\0 & 0 & 0 & -\rmi \\0 & 0 & 0 & 0 \\0 & \rmi & 0 & 0
\end{pmatrix},
\]
\[
G_{31}=\frac{1}{\sqrt{2}}\begin{pmatrix}0 & 0 & 1 & 0 \\0 & 0 & 0 & 0 \\1 & 0 & 0 & 0 \\0 & 0 & 0 & 0
\end{pmatrix}, \quad
G_{32}=\frac{1}{\sqrt{2}}\begin{pmatrix}0 & 0 & 0 & 0 \\0 & 0 & 1 & 0 \\0 & 1 & 0 & 0 \\0 & 0 & 0 & 0
\end{pmatrix}, \]
\[G_{33}=\frac{1}{\sqrt{2}}\begin{pmatrix}0 & 0 & 0 & 0 \\0 & 0 & 0 & 0 \\0 & 0 & 0 & -\rmi \\0 & 0 & \rmi & 0
\end{pmatrix},	\quad 
G_{41}=\frac{1}{\sqrt{2}}\begin{pmatrix}0 & 0 & 0 & 1 \\0 & 0 & 0 & 0 \\0 & 0 & 0 & 0 \\1 & 0 & 0 & 0
\end{pmatrix}, \]
\[G_{42}=\frac{1}{\sqrt{2}}\begin{pmatrix}0 & 0 & 0 & 0 \\0 & 0 & 0 & 1 \\0 & 0 & 0 & 0 \\0 & 1 & 0 & 0
\end{pmatrix}, \quad
G_{43}=\frac{1}{\sqrt{2}}\begin{pmatrix}0 & 0 & 0 & 0 \\0 & 0 & 0 & 0 \\0 & 0 & 0 & 1 \\0 & 0 & 1 & 0
\end{pmatrix},
\]
\[
G_{51}=\frac{1}{\sqrt{2}}\begin{pmatrix}1 & 0 & 0 & 0 \\0 & -1 & 0 & 0 \\0 & 0 & 0 & 0 \\0 & 0 & 0 & 0
\end{pmatrix},\quad
G_{52}=\frac{1}{\sqrt{6}}\begin{pmatrix}1 & 0 & 0 & 0 \\0 & 1 & 0 & 0 \\0 & 0 & -2 & 0 \\0 & 0 & 0 & 0
\end{pmatrix},
\]
\[
G_{53} = \frac{1}{2\sqrt{3}}\begin{pmatrix}1 & 0 & 0 & 0 \\0 & 1 & 0 & 0 \\0 & 0 & 1 & 0 \\0 & 0 & 0 & -3
\end{pmatrix}.
\]


\end{document}